\documentclass[12pt]{llncs}
\usepackage{amssymb,amsmath,color,tikz,subfig}

\usepackage[margin=1.25in]{geometry}
\usetikzlibrary{through,calc,patterns}

\def\bx{\vec x}\def\by{\vec y}
\def\bp{\vec p}
\def\R{\mathbf R}
\def\ne{\nearrow}\def\se{\searrow}
\def\nw{\nwarrow}\def\sw{\swarrow}
\def\sne{{\nearrow\!\!\!\nearrow}}\def\sse{{\searrow\!\!\!\searrow}}
\def\snw{{\nwarrow\!\!\!\nwarrow}}\def\ssw{{\swarrow\!\!\!\swarrow}}

\def\oriendisc #1#2 {%
\begin{scope}
	\draw
		let \p1 = ($ #2 - #1 $), \n1 = {veclen(\x1,\y1)/2} in
		[dashed] ($#1!.5!#2$) circle (\n1);
	\clip
		let \p1 = ($ #2 - #1 $), \n1 = {veclen(\x1,\y1)/1.5} in
		#1 circle (\n1);
	\fill[color=black] #1 circle (0.35);
	\clip
		let \p1 = ($ #2 - #1 $), \n1 = {veclen(\x1,\y1)/2} in
		[preaction={draw=black,thick}] ($#1!.5!#2$) circle (\n1);
\end{scope}
}

\numberwithin{theorem}{section}
\numberwithin{lemma}{section}
\numberwithin{corollary}{section}

\title{Testing Consumer Rationality using\\ Perfect Graphs and Oriented Discs}
\titlerunning{Testing Consumer Rationality}
\toctitle{Testing Consumer Rationality using Perfect Graphs and Oriented Discs}
\author{Shant Boodaghians\inst{1} \and Adrian Vetta\inst{2}}
\institute{Department of Mathematics and Statistics, McGill University\\
\email{shant.boodaghians@mail.mcgill.ca}
\and Department of Mathematics and Statistics, and School of Computer Science, McGill University
\email{vetta@math.mcgill.ca}}

\begin{document}

\maketitle

\addtocounter{footnote}{-2}

\begin{abstract}
Given a consumer data-set, the axioms of revealed preference proffer a binary test for 
rational behaviour. A natural (non-binary) measure of the degree of rationality exhibited 
by the consumer is the minimum number of data points whose removal 
induces a rationalisable data-set. We study the computational complexity of the
resultant {\sc consumer rationality} problem in this paper.
This problem is, in the worst case, equivalent (in terms of approximation) to the {\em directed feedback vertex set} problem.
Our main result is to obtain an exact threshold on the number of commodities 
that separates easy cases and hard cases. 
Specifically, for two-commodity markets  
the {\sc consumer rationality} problem is polynomial time solvable;
we prove this via a reduction to the {\em vertex cover} problem on perfect graphs.
For three-commodity markets, however, the problem is NP-complete; we prove this
using a reduction from {\sc planar 3-sat} that is based upon oriented-disc drawings.
\end{abstract}



\section{Introduction}
The theory of revealed preference, introduced by Samuelson \cite{Sam38,Sam48}, 
has long been used in economics to test for rational behaviour. 
Specifically, given a set of $m$ commodities with price vector $\bp$, we wish 
to determine whether the consumer always demands an affordable bundle $\bx$ of 
maximum utility.
To test this question, assume we are given a collection of {\em consumer data} 
$\{(\bp_1,\bx_1), (\bp_2,\bx_2), \dots, (\bp_m,\bx_m)\}$.
Each pair $(\bp_i,\bx_i)$ denotes the fact that the consumer purchased 
the bundle of goods $\bx_i\in \R^n$ when the prices were $\bp_i \in \R^n$.
(Here $\R=\mathbb{R}_{\ge 0}$ denotes the set of non-negative real numbers.)
Now, assuming the consumer is rational, the selection of $\bx_i$ reveals 
information about the consumer's preferences; 
in particular, suppose that ${\bp_i\cdot\bx_i\geq \bp_i\cdot\bx_j}$ for some
$j\neq i$.
This means that the bundle $\bx_j$ was affordable, and available for selection, 
when $\bx_i$ was chosen.
In this case, we say $\bx_i$ is \textit{directly revealed preferred} to 
$\bx_j$ and denote this $\bx_i\succeq\bx_j$. 
Furthermore, suppose we observe that $\bx_i\succeq \bx_j$ and that $\bx_j\succeq\bx_k$.
Then, by transitivity of preference, we say $\bx_i$ is 
\textit{indirectly revealed preferred} to $\bx_k$. 

For clarity of presentation, we will assume that all the chosen bundles are 
distinct and that all revealed preferences are strict (no ties).
For a rational consumer, the data-set should then have the following property:\\

\vspace{1em}
\noindent\textbf{The Generalized Axiom of Revealed Preference.}
\footnote{When ties are possible, this formulation is called the 
\textit{strong axiom of revealed preference}; see Houthakker~\cite{Hou50}. We refer the reader to 
the survey by Varian \cite{Var05} for details concerning the assorted axioms of revealed preference.}\\
If $\vec x_1\succeq \vec x_2$, 
$\vec x_2 \succeq \vec x_3,\, \dotsc,\, \vec x_{k-2}\succeq \vec x_{k-1}$ and $\vec x_{k-1}\succeq \vec x_k$
then $\vec x_k\nsucceq \vec x_1$.
\vspace{1em}

Moreover, Afriat~\cite{Afr67} showed that the Generalized Axiom of Revealed
Preference ({\sc garp}) is also sufficient for the construction of a 
utility function which \textit{rationalises} the data-set. 
That is, Afriat showed that if the consumer data satisfies {\sc garp}
then one can construct a utility function $v:\R^n\to\R$ 
such that $v$ is maximised at $\bx_i$ among the set of affordable 
bundles at prices $\bp_i$. 
Hence, {\sc garp} is a necessary and sufficient condition for 
consumer rationality. 

We can represent the preferences revealed by the consumer data via a 
directed graph, $D_\succeq =(V,A)$. 
This directed {\em revealed preference graph} contains a vertex $\bx_i\in V$ 
for each data-pair $(\bp_i,\bx_i)$, and an arc from $\bx_i$ to $\bx_j$ if and 
only if $\bx_i\succeq \bx_j$. 
Observe that {\sc garp} holds {\em if and only if} the revealed preference graph 
is acyclic. 
Consequently, Afriat's theorem implies that the consumer is rational if and only 
if $D_\succeq$ contains no directed cycles. 

For example, Figure~\ref{fig:pref-example} displays visually two sets of consumer data. 
Each bundle $\bx_i$ is paired with its price vector $\bp_i$, and a dotted line is drawn through $\bx_i$ 
perpendicular to $\bp_i$. Note that $\bp_i\bx_i\geq \bp_i\by$ if and only if $\by$ lies on the 
opposite side of the dotted line to the drawing of $\bp_i$. Hence, for the first consumer (left), we have 
$\bx_3\succeq\bx_2$, $\bx_3\succeq \bx_1$ and $\bx_2 \succeq \bx_1$. This produces an acyclic 
revealed preference graph $D_\succeq$ and, therefore, her behaviour can be rationalized. On the otherhand, the 
second consumer (right) reveals $\bx_3\succeq\bx_2\succeq \bx_3$. This produces a directed $2$-cycle in $D_\succeq$ and, so, 
her behaviour cannot be rationalised.
\begin{figure}[h]
\centering
{
	\begin{tikzpicture}[scale=0.5]
	\clip[draw] (-3,-2) rectangle (6,5);
	\begin{scope}[rotate=-22.5]
		\coordinate (A) at (0,0);
		\coordinate (B) at (2,1.5);
		\coordinate (C) at (2,4);
		\draw[thick, -stealth] (A)-- ++(90:1);
		\draw[thick, -stealth] (B)-- ++(45:1);
		\draw[thick, -stealth] (C)-- ++(90:1);
		\fill[black] (A) circle (0.14);
		\fill[black] (B) circle (0.14);
		\fill[black] (C) circle (0.14);
		\node at (A) [below left] {$\bx_1$};
		\node at (B) [below left] {$\bx_2$};
		\node at (C) [below left] {$\bx_3$};
		\node at ($(A)+(90:1)$) [left] {$\bp_1$};
		\node at ($(B)+(45:1)+(.2,0)$) [below] {$\bp_2$};
		\node at ($(C)+(90:1)$) [left] {$\bp_3$};
		\draw[dashed] (A) ++(-10,0) -- ++(20,0);
		\draw[dashed] (B) ++(135:10)-- ++(-45:20);
		\draw[dashed] (C) ++(-10,0) -- ++(20,0);
	\end{scope}
	\end{tikzpicture}
}
\qquad
{
	\begin{tikzpicture}[scale=0.5]
	\clip[draw] (-3,-2) rectangle (6,5);
	\begin{scope}[rotate=-22.5]
		\coordinate (A) at (-1,-1);
		\coordinate (B) at (3,1.5);
		\coordinate (C) at (-1,3);
		\draw[thick, -stealth] (A)-- ++(90:1);
		\draw[thick, -stealth] (B)-- ++(45:1);
		\draw[thick, -stealth] (C)-- ++(90:1);
		\fill[black] (A) circle (0.14);
		\fill[black] (B) circle (0.14);
		\fill[black] (C) circle (0.14);
		\node at (A) [below left] {$\bx_1$};
		\node at (B) [below left] {$\bx_2$};
		\node at (C) [below left] {$\bx_3$};
		\node at ($(A)+(90:1)$) [left] {$\bp_1$};
		\node at ($(B)+(45:1)+(.2,0)$) [below] {$\bp_2$};
		\node at ($(C)+(90:1)$) [left] {$\bp_3$};
		\draw[dashed] (A) ++(-10,0) -- ++(20,0);
		\draw[dashed] (B) ++(135:10)-- ++(-45:20);
		\draw[dashed] (C) ++(-10,0) -- ++(20,0);
	\end{scope}
	\end{tikzpicture}
}
\caption{A rational consumer and an irrational consumer.}
\label{fig:pref-example}
\end{figure}
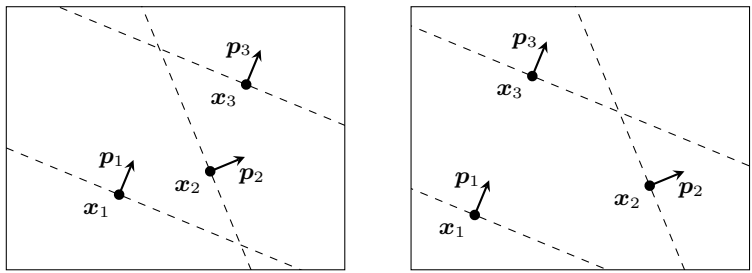 

\vspace*{-2em}

\subsection{A Measure of Consumer Rationality}
We have seen that graph acyclicity can be used to provide a test for 
consumer rationality.
However such a test is binary and, in practice, leads  to the 
immediate conclusion of irrationality, as observed data sets typically induce cycles in 
the revealed preference graph.
Consequently, there has been a large body of experimental and
theoretical work designed to measure how close to rational
the behaviour of a consumer is. 
Examples include measurements based upon best-fit perturbation errors 
($e.g.$ Afriat~\cite{Afr73} and Varian~\cite{Var90}),
measurements based upon counting the number of rationality violations 
present in the data 
($e.g.$ Swofford and Whitney~\cite{SW87} and Famulari~\cite{Fam95}), 
and measurements based upon the maximum size of a {\em rational} subset of 
the data ($e.g.$ Koo~\cite{Koo63} and Houtman and Maks~\cite{HM85}).
Gross~\cite{Gro95} provides a review and analysis of some of these measures.
Recently new measures have been designed by Echenique et al.~\cite{ELS11}, 
Apesteguia and Ballester~\cite{AB15}, and Dean and Martin~\cite{DM15}.

Combinatorially, perhaps the most natural measure is simply to count the 
number of ``irrational" purchases. 
That is, what is the minimum number of data-points whose removal induces a
rational set of data? The associated decision problem is called the 
{\sc consumer rationality} problem.

\vspace{1em}
\noindent\textbf{CONSUMER RATIONALITY}\\
\textbf{Instance:} Consumer data 
	$(\bp_1,\bx_1),\,\dotsc,\,(\bp_m,\bx_m)\in\R^n\times\R^n$, 
and an integer $k$.\\
\textbf{Problem:} \parbox[t]{10.4cm}{Is there a sub-collection of at 
most $k$ data points whose removal produces a data set satisfying {\sc garp}?}
\vspace{1em}

We note that this {\sc consumer rationality} problem is dual to the 
measure of Houtman and Maks~\cite{HM85}.
Using the graphical representation, it can be seen that the consumer 
rationality problem is a special case of the {\sc directed feedback vertex 
set} problem. 
In fact, as we explain in Section~\ref{sec:many}, when there are many goods,
the two problems are equivalent.
However, the consumer rationality problem becomes easier to 
approximate as the number of commodities falls. Indeed, the main contribution 
of this paper is to obtain an exact threshold on the number of commodities 
that separates easy cases (polynomial) and hard cases (NP-complete). 
In particular, we prove the problem is polytime solvable for a two-commodity
market (Section \ref{2comm}), but that it is NP-complete for a three-commodity
market (Section \ref{3comm}).

\section{The General Case: Many Commodities}\label{sec:many}
In this section we show that the {\sc consumer rationality} problem in full 
generality is computationally equivalent to the {\sc directed feedback vertex 
set} ({\sc dfvs}) problem.

\vspace{1em}
\noindent\textbf{DIRECTED FEEDBACK VERTEX SET}\\
\textbf{Instance:} A directed graph $D=(V,A)$, and an integer $k$.\\
\textbf{Problem:} \parbox[t]{10.4cm}{Is there a set $S$ of at most $k$ 
vertices such that the induced subgraph $D[V\setminus S]$ is acyclic?
(Such a set $S$ is called a \textit{feedback vertex set}.)}
\vspace{1em}

First, observe that the {\sc consumer rationality problem} is a special case 
of the {\sc directed feedback vertex set problem}:
we have seen that the dataset is rationalizable if and only if the preference 
graph is acyclic.
Thus, the minimum feedback vertex set
in the preference graph $D_\succeq$ clearly corresponds to the minimum
number of data points that must be removed to create a rationalizable data-set.

On the other hand, provided the number of commodities is large, {\sc dfvs} 
is a special case of the {\sc consumer rationality problem}.
Specifically, Deb and Pai~\cite{DP14} show that for any directed graph $D$ 
there is a data-set on $m=n$ commodities whose preference graph is $D_\succeq 
= D$; for completeness, we include the short proof of this result.

\begin{lemma}\label{lem:many}\cite{DP14} 
	Given sufficiently many commodities, we may construct any digraph 
	as a preference graph.
\end{lemma}
\begin{proof}
	Let $D$ be any digraph on $n$ nodes. We will construct $n$ pairs 
	in $\R^n\times\R^n$ such that $D_\succeq\cong D$. 
	Denote $\bp^i=(p^i_1,\,\dotsc,\,p^i_n)$, and set $p^i_i=1$, $p^i_j=0$ 
	for $j\neq i$. 
	Similarly, denote $\bx^i=(x^i_1,\,\dotsc,\,x^i_n)$, and set 
	$x^i_j = 1$ if $i=j$, 
	0 if $(i,j)\in D$, 
	and 2 if $(i,j)\notin D$.	
	We then have, $\bp_i\cdot \bx_i = 1$, $\bp_i\cdot \bx_j=0$ if we 
	want an arc from $i$ to $j$, and $\bp_i\cdot \bx_j=2$ if we do not 
	want an arc, as desired. \qed
\end{proof}

It follows that any lower and upper bounds on approximation for
(the optimization version of) {\sc dfvs} immediately apply to 
(the optimization version of) the {\sc consumer rationality} problem. 
The exact hardness of approximation for {\sc dfvs} is not known.
The best upper bound is due to Seymour~\cite{Sey95} who gave
an $O(\log n \log\log n)$ approximation algorithm.
With respect to lower bounds, the {\sc directed feedback vertex set} problem 
is NP-complete~\cite{Kar72}.
Furthermore, as we will see in Section~\ref{2comm}, the 
{\sc consumer rationality} problem is at least as hard to approximate as 
{\sc vertex cover}. 
It follows that {\sc dfvs} problem cannot be approximated to within a 
factor $1.36$ \cite{DS05} unless $P=NP$.
Also, assuming the Unique Games Conjecture \cite{Kho02}, the minimum directed feedback vertex set 
cannot be approximated to within any constant factor \cite{GHM11,Sve12}.

Lemma \ref{lem:many} shows the equivalence with {\sc directed feedback 
vertex set} applies when the number of commodities is at least the size of 
the data-set.
However, Deb and Pai~\cite{DP14} also show that for an \mbox{$m$-commodity} 
market, there exists a directed graph on $O(2^m)$ vertices that cannot be 
realised as a preference graph. 
This suggests that the hardness of the consumer rationality problem may vary 
with the quantity of goods. Indeed, we now prove that this is the case.



\section{The Case of Two Commodities} \label{2comm}
We begin by outlining the basic approach to proving polynomial solvability for two goods.
As described, the {\sc consumer rationality} problem is a special case of {\sc dvfs}.
For two goods, however, rather than considering all 
directed cycles, it is sufficient to find a vertex hitting set for the set of 
{\em digons} (directed cycles consisting of two arcs). 
The resulting problem can be solved by finding a minimum vertex cover in a 
corresponding auxiliary undirected graph. 
The {\sc vertex cover} problem is, of course, itself hard \cite{DS05}.
But we prove that the auxiliary undirected graph is perfect, and 
{\sc vertex cover} is polytime solvable in perfect graphs.

\subsection{Two-Commodity Markets and the Vertex Cover Problem}
So, our first step is to show that it suffices to hit only digons.
Specifically, we prove that every vertex-minimal cycle in the revealed 
preference graph $D_{\succeq}$ is a digon. 
This fact corresponds to the result that for two goods the 
\textit{Weak Axiom of Revealed Preference} is equivalent to the 
\textit{Generalised Axiom of Revealed Preference}. This equivalence was noted by
Samuelson~\cite{Sam48} and formally proven by Rose~\cite{Ros58} in 1958; for a 
recent structurally motivated proof see~\cite{Heu14}.
For completeness, and to illustrate some of the notation and techniques required 
in this paper, we present a short geometric proof here.

We begin with the required notation. Let ${\bx=(x_1,x_2)\in\R^2}$, and define 
\[
	{\bx^\se := \{(y_1,y_2)\in\R^2:y_1\geq x_1,\,y_2\leq x_2\}}\enspace,
\] 
$i.e.$ the points which lie ``below and to the right'' of $\bx$ in the plane.
Define $\bx^\nw$, $\bx^\ne$ and $\bx^\sw$ similarly. In addition, define $\bx^\sse$ 
$\bx^\snw$, $\bx^\sne$ and $\bx^\ssw$ by replacing the inequalities with strict
inequalities. 
Furthermore, if $\ell$ is a line in the plane of non-positive slope which 
intersects the positive quadrant, we say a point \textit{lies below} 
$\ell$ if it lies in the same closed half-plane as the origin. For each data pair $(\bp_i, \bx_i)$, we define 
$\ell_i$ to be the line through $\bx_i$ perpendicular to $\bp_i$.
Hence, in our setting $\bx_i\succeq \bx_j$ if and only if 
$\bx_j$ lies below $\ell_i$. Note that, if $\bx_i\succeq\bx_j$, then we may 
not have $\bx_j\in \bx_i^\sne$ since $\bp_i$ is non-negative.

\begin{lemma}\label{lem:warp} \cite{Ros58}
For two commodities, every minimal cycle is a digon.
\end{lemma}

\begin{proof}
Let $C_k=\{\bx_0,\bx_1, \dots ,\bx_{k-1}\}$, listed in order, be a 
vertex-minimal directed cycle in $D_{\succeq}$. 
Suppose, for a contradiction, that $k\ge 3$. 
By minimality, the cycle $C_k$ is chordless, therefore, $\bx_i\succeq \bx_j$ 
if and only if $j=i+1\pmod k$. 
(Henceforth, we will often assume without statement that indices are taken modulo $k$. 
Furthermore, ``left'' will stand for the 
negative $x$ direction.) Without loss of generality, suppose $\bx_i$ is the 
leftmost bundle -- or one of them. Since $\bx_i\succeq \bx_{i+1}$, we have 
that $\bx_{i+1}$ must fall in $\bx_i^\se$. We claim that $\ell_i$ must be
steeper than $\ell_{i+1}$. To see this, suppose this is not true. Then, as 
shown in Figure~\ref{fig:2-comm-digons}(a), $\ell_{i+1}$ must intersect the line 
$\ell_{i}$ strictly to the left of $\bx_i$. 
If not, $\bx_{i+1}\succeq \bx_{i}$. Now $\bx_{i+2}$ lies under $\ell_{i+1}$ but not under $\ell_i$, but this 
implies that $\bx_{i+2}$ lies strictly to the left of $\bx_i$ as illustrated.
This gives the desired contradiction. 
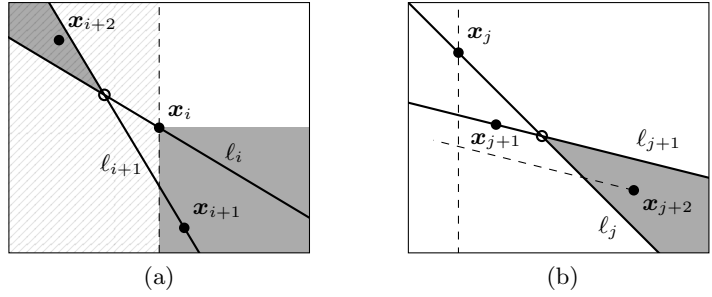
\begin{figure}[h]
\centering
\subfloat[][]
{\centering
	\begin{tikzpicture}[scale=0.333]
		\clip[draw] (-6,-5) rectangle (6,5);
		
		\node at (0,0) [above right] {$\bx_i$};
		\node at (1,-4) [above right] {$\bx_{i+1}$};
		\node at (-4, 3.5) [above right] {$\bx_{i+2}$};
		
		\node at (3,-1) {$\ell_i$};
		\node at (-1.5,-1.5) {$\ell_{i+1}$};
		
		\fill[black] (0,0) circle (0.21);
		\fill[black] (1,-4) circle (0.21);
		\fill[black] (-4,3.5) circle (0.21);
		
		\fill[black, opacity=0.33] (0,0) rectangle (10,-5);
		\fill[black, opacity=0.33, pattern=north east lines] 
			(0,-5) rectangle (-10,5);
		
		\fill[opacity=0.33, black] (-6,6) -- (-10,6) -- (-2.1875,1.3125) 
			-- (-5,6) -- cycle;
	
		\draw[thick] (-10,6)--(10,-6);
		\draw[thick] (4,-9)--(-5,6);
		\draw[dashed] (0,-5)--(0,5);
		
		\draw[thick] (-2.1875,1.3125) circle (0.21);
	\end{tikzpicture}}\qquad\qquad
\subfloat[][]{\centering
	\begin{tikzpicture}[scale=0.333]
		\clip[draw] (-2,-8) rectangle (10,2);
		
		\fill[opacity=0.33, black] (15,-15) -- (3.333,-3.333) 
			-- (16,-6.5) -- (16,-15) -- cycle;
	
		\draw[thick] (-5,5)--(15,-15);
		\draw[thick] (-6,-1)--(14,-6);
		\draw[dashed] (0,-10)--(0,5);
		\draw[dashed] (7,-5.5)--(-1,-3.5);
		
		\draw[thick] (3.333,-3.333) circle (0.20);
		
		\node at (0,0) [above right] {$\bx_j$};
		\node at (1.5,-2.875) [below] {$\bx_{j+1}$};
		\node at (7,-5.5) [below right] {$\bx_{j+2}$};
		
		\fill[black] (0,0) circle (0.21);
		\fill[black] (1.5,-2.875) circle (0.21);
		\fill[black] (7,-5.5) circle (0.21);
		
		\node at (8,-3.5) {$\ell_{j+1}$};
		\node at (6,-7) {$\ell_{j}$};
	\end{tikzpicture}

}
	\caption{Leftmost 2-commodity bundles on a cycle.}
	\label{fig:2-comm-digons}
\end{figure}
Hence, $\ell_{i}$ must be steeper 
than $\ell_{i+1}$. This situation is illustrated in Figure~\ref{fig:2-comm-digons}(b)
where we set $j=i$. We claim the following:
\begin{claim} 
	Suppose $\bx_{j+1}\in\bx_j^\searrow$ and $\ell_{j}$ is steeper than 	
	$\ell_{j+1}$, then we must have that $\bx_{j+2}\in\bx_{j+1}^\searrow$ 
	and that $\ell_{j+1}$ is steeper than $\ell_{j+2}$.
\end{claim} 
As shown in Figure~\ref{fig:2-comm-digons}(b), because $\ell_{j}$ is steeper than $\ell_{j+1}$, we must 
have $\bx_{j+2}\in\bx_{j+1}^\searrow$. 
It remains to show that $\ell_{j+1}$ is steeper than $\ell_{j+2}$. 
Suppose not, then, since $\bx_{j+1}$ must fall above $\ell_{j+2}$, 
the (highlighted) point where $\ell_{j}$ and $\ell_{j+1}$ meet must also 
fall above $\ell_{j+2}$. 
Thus, the region which falls above both $\ell_j$ and $\ell_{j+1}$ cannot 
intersect the region below $\ell_{j+2}$. Therefore,
there is no valid position for $\bx_{j+3}$. 
Consequently, $\ell_{j+1}$ must be steeper than $\ell_{j+2}$, as desired.

Hence, by induction, for every $0\leq j\leq k-1$, we have that $\ell_j$ is 
steeper than $\ell_{j+1}$ and that $\bx_{j+1}\in\bx_j^\searrow$, 
where our base case is $j=i$.
However, this cannot hold for $j=i-1$, since $\bx_i$ is the leftmost 
point in the cycle, amounting to a contradiction, 
and refuting the assumption that there existed a minimal cycle on 
at least 3 vertices.\qed
\end{proof}

Lemma~\ref{lem:warp} implies that a vertex set that intersects every digon 
will also intersect each directed cycle of any length. 
Hence, to solve the {\sc consumer rationality} problem for two goods, 
it suffices to find a minimum cardinality hitting vertex set for the digons 
of $D_\succeq$. 
We can do this by transforming the problem into one of finding a minimum 
vertex cover in an undirected graph. Recall the {\sc vertex cover} problem is: 

\vspace{1.5em}
\noindent\textbf{VERTEX COVER}\\
\textsc{instance:} Given an undirected graph $G=(V,E)$ and an integer~$k$.\\
\textsc{problem:} \parbox[t]{10.4cm}{Is there a set $S$ of at most $k$ 
vertices such that every edge has an endpoint in $S$?}
\vspace{1em}

The transformation is then as follows: given the directed revealed preference 
graph $D_\succeq$ we create an {\em auxiliary undirected graph} $G_\succeq$. 
The vertex set $V(G_\succeq)=V(D_\succeq)$ so the undirected graph also has 
a vertex for each bundle $\bx_i$.
There is an edge $(\bx_i,\bx_j)$ in $G_\succeq$  if and only if $\bx_i$ and 
$\bx_j$ induce a digon in $D_\succeq$. 
It is easy to verify that a vertex cover in $G_\succeq$ corresponds to 
a hitting set for digons of $D_\succeq$. 

Let's see some simple examples for the auxiliary graph $G_\succeq$.
First consider Figure \ref{fig:2-comm-constr}(a), where bundles are 
placed on a concave curve.
Now every pair of vertices $\bx_i$ and $\bx_j$ induce a digon in $D_\succeq$.
Thus $G_\succeq$ is an undirected clique. 
Now consider Figure \ref{fig:2-comm-constr}(b).
The vertices on the left induce a directed path in $D_\succeq$; 
the vertices along the bottom also induce a directed path in $D_\succeq$. 
However each pair consisting of one vertex on the left and one vertex on
the bottom induce a digon in $D_\succeq$. Thus $G_\succeq$ is a 
complete bipartite graph.

\begin{figure}[h]
	\centering
	\subfloat[][A Complete Graph]{
		\begin{tikzpicture}
			\clip (0,3.25) rectangle (3.25,0);
			\foreach \r in {0,1,...,5} {
			\begin{scope}[rotate=-15*\r - 7.5]
				\fill[black] (0,3) circle (0.07);
				\draw[thick] (-5,3)--(5,3);
			\end{scope}
			}
	\end{tikzpicture}
	} \qquad\quad\quad \subfloat[][A Complete Bipartite Graph]{
		\begin{tikzpicture}[scale=0.5]
			\clip (-3,6) rectangle (7,-0.5);
			\foreach \s in {1,2,...,5} {
				\fill[black] (\s,0) circle (0.14);
				\fill[black] (-.25*\s,\s) circle (0.14);
				
				\draw[thick] (\s+0.25,-1)--(\s-1.75,7);
				\draw[thick] (-1-0.25*\s,\s)--(6.333-0.25*\s,\s);
			}
	\end{tikzpicture}
%
}
\caption{Examples of the Auxiliary Undirected Graph.}
\label{fig:2-comm-constr}
\end{figure}
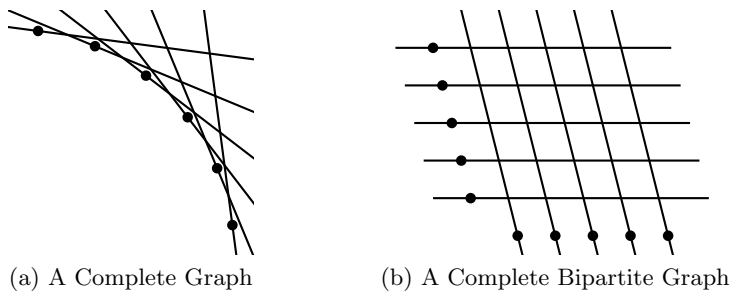 

\subsection{Perfect Graphs}
An undirected graph $G$ is perfect if the chromatic number of any induced
subgraph is equal to the cardinality of the maximum clique in the subgraph. 
In 1961, Berge~\cite{Ber61} made the famous conjecture that an undirected 
graph is perfect if and only if it contains neither an odd length hole 
nor an odd length antihole.
Here a {\em hole} is a chordless cycle with at least four vertices.  
An {\em antihole} is the complement of a chordless cycle with at least four vertices.
Berge's conjecture was finally proven by Chudnovsky, Robertson, Seymour and Thomas \cite{CRST06} in 2006.

\begin{theorem}[The Strong Perfect Graph Theorem \cite{CRST06}]
An undirected graph is perfect if and only if it contains no odd holes 
and no odd antiholes.
\end{theorem}

There are many important classes of perfect graphs, for example, cliques, 
bipartitie graphs, chordal graphs, line graphs of bipartite graphs, and
comparability graphs.\footnote{
	By the (Weak) Perfect Graph Theorem \cite{Lov72}, the complements of 
	these classes of graphs are also perfect.} 
Interestingly, we now show that the class of 2D auxiliary revealed 
preference graphs are also perfect.
To prove this, we will need the following geometric lemma, but first, we
introduce the required notation. 


\begin{lemma}\label{lem:induced path}
Let $\{x_i, x_j, x_k\}$, listed in order, be an induced path in the 
2D auxiliary revealed preference graph $G_\succeq$.
If $\bx_i\in \bx_j^\nw$ then $\bx_k\in\bx_j^\nw$. 
(Similarly, if  $\bx_i\in \bx_j^\se$ then $\bx_k\in\bx_j^\se$.) 
\end{lemma}
\begin{proof}
Recall the assumption that the bundles distinct, that is, $\bx_i\neq\bx_j$ for all $i\neq j$.
Because $\{x_i, x_j\}$ is an edge in the auxiliary undirected graph $G_\succeq$, 
we know that ${\bx_i\succeq\bx_j}$ and ${\bx_j\succeq\bx_i}$. 
Therefore it cannot be the case that $\bx_i\in \bx_j^\sne$ or 
${\bx_j\in\bx_i^\sne}$. 
Thus, either $\bx_j\in\bx_i^\nw$ or $\bx_j\in\bx_i^\se$, but not both.
Similarly, because $\{x_j, x_k\}$ is an edge in $G_\succeq$, either 
$\bx_k\in\bx_j^\nw$ or $\bx_k\in\bx_j^\se$.

Now, without loss of generality, let $\bx_i\in\bx_j^\nw$. 
For a contradiction, assume that $\bx_k\in\bx_j^\se$.
Hence, we have ${\bx_j\in\bx_i^\se\cap\bx_k^\nw}$. 
Suppose $\bx_j$ lies strictly below the line~$\ell_{i,k}$ through 
$\bx_i$ and $\bx_k$. 
But then we cannot have both $\bx_j\succeq \bx_i$ and $\bx_j\succeq\bx_k$.
This is because the line $\ell_j$ must cross the segment of $\ell_{i,k}$
between $\bx_i$ and $\bx_k$ if it is to induce either of the two preferences.
Thus, the line $\ell_j$ separates $\bx_i$ and $\bx_k$ and, so, at most one of bundles can lie below the line.
This is illustrated in Figure~\ref{fig:2-comm-induced}(a).

	 \begin{figure}[h]
	\centering
	\subfloat[][$\bx_j$ below $\ell_{i,k}$]{
		\begin{tikzpicture}[scale=0.5]
			\node at (0,4) [above right] {$\bx_i$};
			\node at (2,1) [below left] {$\bx_j$};
			\node at (5,0) [above right] {$\bx_k$};
			
			\node at (-.5,1.5) {$\ell_j$};			
			\node at (3,-.5) {$\ell_j'$};
			\node at (-1,4) {$\ell_{i,k}$};		
			
			\draw[dashed] (-1,4.8)--(6,-0.8);
			\fill[black] (0,4) circle (0.14);
			\fill[black] (5,0) circle (0.14);	
			\fill[black] (2,1) circle (0.14);	

			\draw[thick] (2.5,-1)--(1,5);
			\draw[thick] (-1,1)--(6,1);
			\fill[black, opacity=0.25] (0,0)--(0,4)--(5,0)--cycle;
	\end{tikzpicture}
	} \qquad\quad\quad \subfloat[][$\bx_j$ above or on $\ell_{i,k}$]{
		\begin{tikzpicture}[scale=0.5]
			\node at (0,4) [below left] {$\bx_i$};
			\node at (3,2.5) [left] {$\bx_j$};
			\node at (5,0) [below left] {$\bx_k$};
			
			\node at (6,3) {$\ell_i$};
			\node at (3.5,4.5) {$\ell_k$};
			\node at (3,1) {$\ell_{i,k}$};
			
			\draw[dashed] (-1,4.8)--(6,-0.8);
			\fill[black] (0,4) circle (0.14);
			\fill[black] (5,0) circle (0.14);	
			\fill[black] (3,2.5) circle (0.14);	

			\fill[black, opacity=0.25] (5,4)--(0,4)--(5,0)--cycle;
			\draw[thick] (5.5,-1) -- (2.5,5);
			\draw[thick] (-1,4.25) -- (6,2.5);
	\end{tikzpicture}
}
	\caption{Induced path on three vertices.}
	\label{fig:2-comm-induced}
\end{figure}
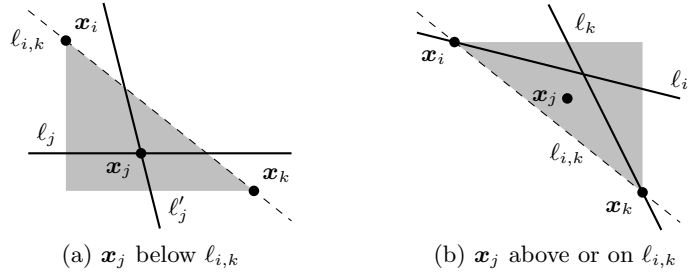 

On the other hand, suppose $\bx_j$ lies on or above the line~$\ell_{i,k}$ 
through $\bx_i$ and $\bx_k$.
Now we know that $\bx_i\succeq\bx_j$. This implies that $\bx_i\succeq \bx_k$, 
as illustrated in Figure~\ref{fig:2-comm-induced}(b).
Furthermore, we know that $\bx_k\succeq\bx_j$ which implies 
that $\bx_k\succeq\bx_i$. 
Thus $\{x_i, x_k\}$ is an edge in $G_\succeq$. 
This contradicts the fact that $\{x_i, x_j, x_k\}$ is an induced path.
\qed
\end{proof}

\begin{lemma}\label{lem:hole}
The 2D auxiliary revealed preference graph $G_\succeq$ contains no
odd holes on at least 5 vertices.
\end{lemma}
\begin{proof} 
Take a hole $C_k=\{\bx_0,\bx_1,\dots, \bx_{k-1}\}$, listed in order, 
where $k\ge 5$ is odd. 
For any $0\le i\le k-1$, the three vertices $\{\bx_{i-1},\bx_i, \bx_{i+1}\}$ 
induce a path in $G_\succeq$. 
Consequently, by Lemma~\ref{lem:induced path}, 
either both $\bx_{i-1}$ and $\bx_{i+1}$ are in $\bx_i^\nwarrow$ or 
 both $\bx_{i-1}$ and $\bx_{i+1}$ are in $\bx_i^\searrow$.
In the former case, colour $\bx_i$ yellow. 
In the latter case, colour $\bx_i$ red.
Thus we obtain a $2$-coloring of $C_k$. 
Since $k$ is odd, there must be two adjacent vertices, 
$\bx_i$ and $\bx_{i+1}$, with the same colour. 
Without loss of generality, let both vertices be yellow. 
Thus, $\bx_{i+1}$ is $\bx_i^\nwarrow$ and $\bx_i$ is in $\bx_{i+1}^\nwarrow$. 
This contradicts the distinctness of $\bx_i$ and $\bx_{i+1}$.
\qed
\end{proof}

We remark that the parity condition in Lemma \ref{lem:hole} is necessary. 
To see this consider the example in Figure~\ref{fig:2-comm-C6} which 
produces an even hole on six vertices. 
Specifically, the only mutually adjacent pairs are the $(\bx_i,\bx_{i+1})$ pairs, 
with indices taken modulo~6. 

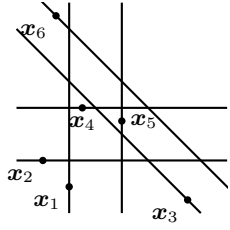
\begin{figure}[h]
	\centering
	\begin{tikzpicture}[scale=0.35]
		\node at (1,0) [below left] {$\bx_1$};
		\node at (0,1) [below left] {$\bx_2$};
		\node at (5.5,-.5) [below left] {$\bx_3$};
		\node at (1.5,3) [below] {$\bx_4$};
		\node at (3,2.5) [right] {$\bx_5$};
		\node at (0.5,6.5) [below left] {$\bx_6$};

		\clip (-1,7) rectangle (7,-1);
		\fill[black] (1,0) circle (0.14);
		\fill[black] (0,1) circle (0.14);
		\fill[black] (5.5,-0.5) circle (0.14);
		\fill[black] (1.5,3) circle (0.14);
		\fill[black] (3,2.5) circle (0.14);
		\fill[black] (0.5,6.5) circle (0.14);
		
		\draw[thick] (1,-1)--(1,7);
		\draw[thick] (3,-1)--(3,7);
		\draw[thick] (-1,1)--(7,1);
		\draw[thick] (-1,3)--(7,3);
		\draw[thick] (-1,6)--(6,-1);
		\draw[thick] (-1,8)--(8,-1);
	\end{tikzpicture}
\caption{Construction for $C_6$}
\label{fig:2-comm-C6}
\end{figure} 

\newpage
\begin{lemma}\label{lem:antihole}
The 2D auxiliary revealed preference graph $G_\succeq$ contains no
antiholes on at least 5 vertices.
\end{lemma}
\begin{proof}
Note that the complement of an odd hole on five vertices is also an odd hole.
Thus, by Lemma \ref{lem:hole}, the graph $G_\succeq$ may not contain an
antihole on five vertices.  

Next consider an antihole $\bar{C}_k=\{\bx_0,\bx_1,\dots, \bx_{k-1}\}$, 
listed in order, with ${k\geq 6}$. 
The neighbours in $\bar{C}_k$ of $\bx_i$, for any $0\le i\le k-1$, 
are $\Gamma_i=\{\bx_{i+2},\bx_{i+3},\dots, \bx_{i-2}\}$.
We claim that either every vertex of $\Gamma_i$ is in $\bx_i^\nwarrow$ or
every vertex of $\Gamma_i$ is in $\bx_i^\searrow$. 
To see this note that $(\bx_{i+2}, \bx_{i+3})$ is not an edge, and 
therefore $\{\bx_{i+2},\bx_{i}, \bx_{i+3}\}$ is an induced path in $G_\succeq$. 
By Lemma~\ref{lem:induced path}, without loss of generality, 
both $\bx_{i+2}$ and $\bx_{i+3}$ are in $\bx_i^\nwarrow$. 
But $\{\bx_{i+3},\bx_{i}, \bx_{i+4}\}$ is also an induced path in $G_\succeq$. 
Consequently, as $\bx_{i+3}$ is in $\bx_i^\nwarrow$, 
Lemma~\ref{lem:induced path} implies that $\bx_{i+4}$ is in $\bx_i^\nwarrow$. 
Repeating this argument through to the induced path 
$\{\bx_{i-3},\bx_{i}, \bx_{i-2}\}$ gives the claim.

Now consider the three vertices $\bx_0, \bx_2$ and $\bx_4$. 
Since $k\ge 6$ these vertices are pairwise adjacent in $\bar{C}_k$. 
Without loss of generality, by the claim, $\Gamma_0$ is in $\bx_0^\nwarrow$. 
Thus, $\bx_2$ and $\bx_4$ are in $\bx_0^\nwarrow$.
However $\bx_0$ is in $\Gamma_2 \cap \Gamma_4$. 
Thus every vertex in $\Gamma_2$ is in $\bx_2^\searrow$ 
and every vertex in  $\Gamma_2$ is in $\bx_4^\searrow$. 
Hence, $\bx_4$ is in $\bx_2^\searrow$ and $\bx_2$ is in $\bx_4^\searrow$, 
a contradiction.
\qed
\end{proof}

Lemmas \ref{lem:hole} and \ref{lem:antihole} together show, 
by applying the Strong Perfect Graph Theorem,
that the auxiliary undirected graph is perfect.
\begin{theorem}\label{thm:perfect}
The 2D auxiliary revealed preference graph $G_\succeq$ is perfect.
\qed 
\end{theorem}

\subsection{A Polynomial Time Algorithm}
In classical work, Gr\"otschel, Lov\'asz and Schrijver~\cite{GLS84,GLS88} 
show that the {\sc vertex cover} problem in a perfect graph
can be solved in polynomial time via the ellipsoid method.

\begin{theorem}\cite{GLS84}
The {\sc vertex cover} problem is solvable in polynomial time on a perfect graph. 
\qed
\end{theorem}
But by Theorem~\ref{thm:perfect}, the auxiliary undirected graph is perfect. 
Since the consumer rationality problem for two commodities corresponds 
to a vertex cover problem on this auxiliary undirected graph, we have:
\begin{theorem}\label{thm:polytime}
In a two-commodity market, the {\sc consumer rationality} problem 
is solvable in polynomial time. \qed
\end{theorem}

\section{The Case of Three Commodities} \label{3comm}
We have shown that for two commodities, 
the consumer rationality problem can be solved in polynomial time. 
We now prove the problem is NP-complete if there are three (or more) commodities
by presenting a reduction from {\sc planar 3-sat}. 
The proof has three parts:
first we transform an instance of {\sc planar 3-sat} to an instance of 
{\sc vertex cover} in an associated undirected {\em gadget graph}. 
Second, we show that a vertex cover in the gadget graph corresponds
to a directed feedback vertex set in a directed {\em oriented disc graph}.
Finally, we prove that every oriented disc graph corresponds to
a preference graph in a three-commodity market. 
Consequently, we can solve {\sc planar 3-sat} using an algorithm for the 
three-commodity case of the {\sc consumer rationality} problem.

We begin by defining the class of oriented-disc graphs. 
Let $\{\bx_1,\,\dotsc,\,\bx_n\}$ be points in the plane and let 
$\{B_1,\,\dotsc,\,B_n\}$ be closed discs of varying radii such that 
$B_i$ contains $\bx_i$ on its boundary. 
We call this collection of points and discs an \textit{oriented-disc drawing.}
Given a drawing, we construct a directed graph $D=(V,A)$ on the vertex 
set $V=\{\bx_1,\dots, \bx_n\}$.
There is an arc from $\bx_i$ to $\bx_j$ in $D$ if $\bx_j$, $j\neq i$, is 
contained in the disc $B_i$. 
A directed graph that can be built in this manner is called an 
{\em oriented-disc graph}.

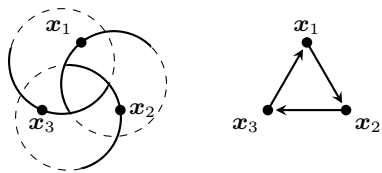
\begin{figure}[h]
\begin{center}\begin{tikzpicture}[scale=0.2]	
	\coordinate (A) at (90:3);
	\coordinate (C) at (210:3);
	\coordinate (B) at (-30:3);

	\node at (A) [above left] {$\bx_1$};
	\node at (B) [right] {$\bx_2$};
	\node at (C) [below] {$\bx_3$};			

	\oriendisc{(A)}{(-30:5)}
	\oriendisc{(B)}{(210:5)}
	\oriendisc{(C)}{(90:5)}
	
	\begin{scope}[shift={(15,0)}]
		\coordinate (A) at (90:3);
		\coordinate (C) at (210:3);
		\coordinate (B) at (-30:3);
		\draw[thick, -stealth] (A)--($(A)!0.9!(B)$);
		\draw[thick, -stealth] (B)--($(B)!0.9!(C)$);
		\draw[thick, -stealth] (C)--($(C)!0.9!(A)$);
		\fill[black] (A) circle (0.35);
		\fill[black] (B) circle (0.35);
		\fill[black] (C) circle (0.35);
		\node at (A) [above] {$\bx_1$};
		\node at (B) [below right] {$\bx_2$};
		\node at (C) [below left] {$\bx_3$};
	\end{scope}
\end{tikzpicture}\end{center}	
\caption{An oriented disc drawing and its corresponding oriented disc graph.}
\label{fig:orien-C3}
\end{figure}

An example is given in Figure~\ref{fig:orien-C3}. 
The oriented-disc drawing is shown on the left and the
the resulting oriented disc graph, a directed cycle on 3 vertices, 
is shown on the right.
(We remark that, for enhanced clarity in the larger figures that follow, the boundary circles are 
drawn half-dotted.) Note that, even if the discs have uniform radii, the resulting oriented-disc 
graphs need not be symmetric -- that is, $(\bx_i, \bx_j)$ can 
be an arc even if $(\bx_j, \bx_i)$ is not. 
This is due to the fact that $\bx_i$ lies on the boundary, not at the centre, 
of its disc $B_i$.
We now start by proving the third part of the reduction: 
every oriented disc graph corresponds to a preference graph in a 
three-commodity market.

\begin{lemma} Every oriented-disc graph corresponds to a preference graph 
induced by consumer data in a three-commodity market.
\end{lemma}
\begin{proof}
Let $D$ be any oriented-disc graph.
We wish to build a three-commodity data set whose preference graph is $D$. 
Recall that the plane is homomorphic to the 2-dimensional sphere minus a point.
Moreover, the inverse of the {\em stereographic projection} is a map from 
the plane to a sphere which preserves the shape of circles; see, for example, \cite{Ear07}. 
This motivates us to attempt to draw the points and discs on the unit 
sphere centered at $(1,1,1)\in\R^3$.
To do this, we scale the oriented-disc drawing appropriately and embed it in a 
small region on the ``underside'' of the sphere, that is, around the point 
where the inwards normal vector is $(1,1,1)$.
An example of this, where the oriented-disc graph is the directed $3$-cycle, is 
shown in Figure \ref{fig:orien-embed}(a). 

We now need to create the corresponding collection of consumer data. 
Let $\{\bx_1,\,\dotsc,\,\bx_n\}$ be the $n$ points of some
oriented-disc drawing of $D$ embedded onto the underside of the sphere. 
Note that the intersection of a sphere and a plane is a circle. 
Furthermore, a plane through a point on the sphere will create a circle 
containing that point.
Thus we may select the $\bx_i$ to be the bundles chosen by the market 
and we may choose $\bp_i$ such that the plane with normal $\bp_i$ that passes
through $\bx_i$ intersects the sphere exactly along the boundary of the embedding 
of the disc $B_i$.
An example is shown in Figure \ref{fig:orien-embed}(b). 
Because $\bp_i$ is non-negative it points into the sphere. 
Therefore, $\bx_i$ is revealed preferred to every point on the inside of the embedding of $B_i$;
it is not revealed preferred to any other point on the sphere.
Hence, the preference graph $D_\succeq$ is isomorphic to the original 
oriented-disc graph, as desired. \qed
\end{proof}

\begin{figure}[h]
	\centering
\subfloat[][]{\centering
	\includegraphics[height=3.5cm]{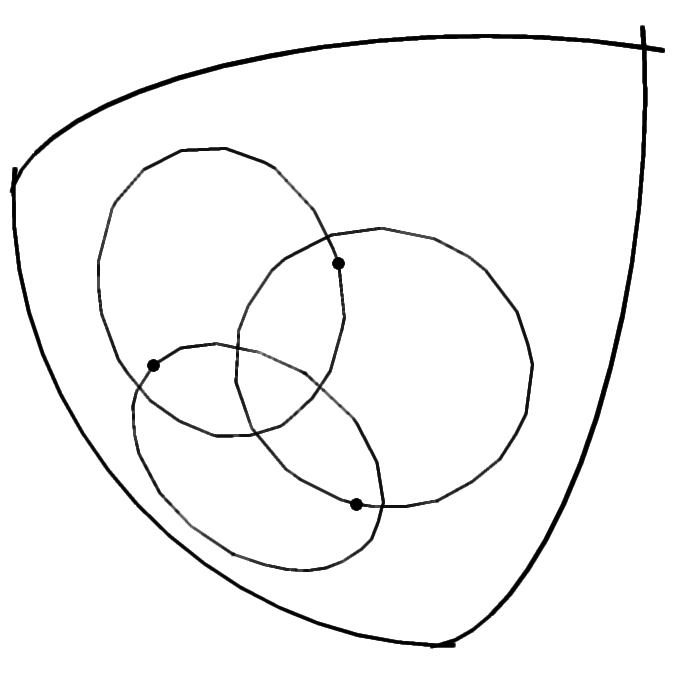}
}\qquad\subfloat[][]{\centering
	\includegraphics[height=3.5cm]{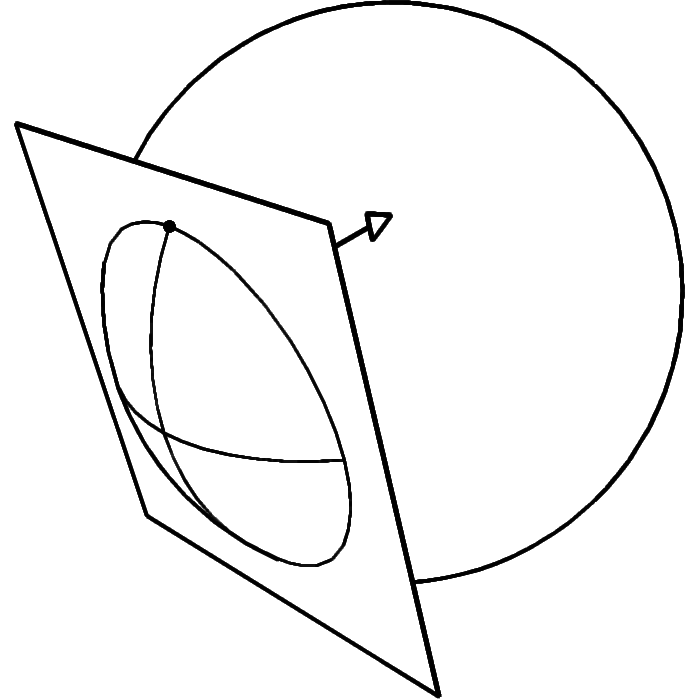}
}
\caption{A 3-cycle embedded on a sphere section, and a disc on a sphere.}
\label{fig:orien-embed}
\end{figure} 

Now, recall the first part of the reduction: 
we wish to transform an instance of {\sc planar 3-sat} to an instance of 
{\sc vertex cover} in an associated undirected \textit{gadget graph}.
Our gadget graph is based upon a network used by Wang and Kuo~\cite{WK88} 
to prove the hardness of {\sc maximum independent set} in undirected 
unit-disc graphs. 
However, we are able to simplify their non-planar network by using an instance of 
{\sc planar 3-sat} rather than the general {\sc 3-sat}. 
This simplification will be useful when implementing the second part of 
the reduction. 

Let $\varphi$ be an instance of {\sc planar 3-sat} with variables 
$u_1,\,\dotsc,\,u_n$ and clauses $C_1,\,\dotsc,\,C_m$. 
Recall that $\varphi$ is \textit{planar} if the bipartite graph 
$H_\varphi$ consisting of a vertex for each variable, a vertex for each 
clause, and edges connecting each clause to its three variables, is planar.
The associated, undirected, \textit{gadget graph} $G_\varphi$ is constructed 
as shown in Figure~\ref{fig:3-sat-graph}.
For each clause $C=(u_i\vee u_j\vee u_k)$, add a $3$-cycle to the graph
whose vertices are labelled by the appropriate literals for the variables 
$u_i$, $u_j$ and $u_k$. 
We call these the \textit{clause gadgets}. For each variable $u_i$, 
add a large cycle of even length whose vertices are alternatingly labelled as 
the literals $u_i$ and $\bar u_i$. We call these the \textit{variable gadgets}. 
Finally, add an edge from each variable in the clause gadgets to some vertex 
on the corresponding variable gadget with the opposite label -- we choose a different variable vertex 
for each clause it is contained in.

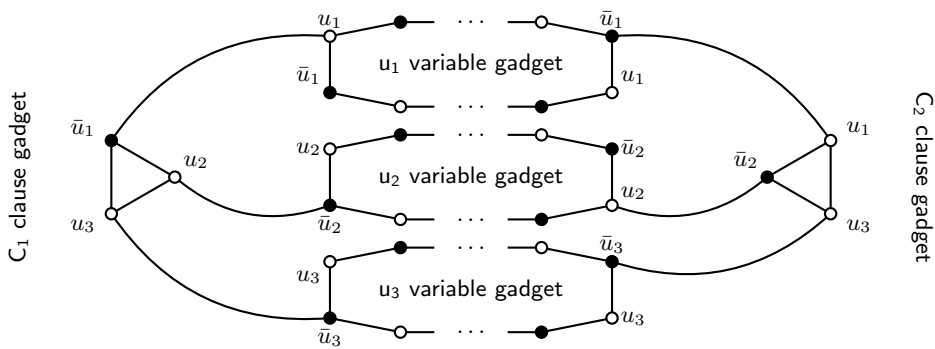
\begin{figure}[ht]
	\centering
	\begin{tikzpicture}[scale=0.75]
		
	\path[thick] ($(-6,.5)+(120:.75)$) edge [bend left=30] ($(0,2)+(-2.5,1)$);
	\path[thick] ($(-6,.5)+(-120:.75)$) edge [bend right=30] ($(0,-2)+(-2.5,0)$);
	\path[thick] ($(-6,.5)+(0:.75)$) edge [bend right=30] ($(0,0)+(-2.5,0)$);
	
	\path[thick] ($(6,.5)+(60:.75)$) edge [bend right=30] ($(0,2)+(2.5,1)$);
	\path[thick] ($(6,.5)+(-60:.75)$) edge [bend left=30] ($(0,-2)+(2.5,1)$);
	\path[thick] ($(6,.5)+(180:.75)$) edge [bend left=30] ($(0,0)+(2.5,0)$);	
	
	\foreach \s in {-2,0,2} {
		\begin{scope}[shift={(0,\s)}]
			\draw[thick] (-.65,1.25)--(-1.25,1.25)--
				(-2.5,1)--(-2.5,0)--(-1.25,-.25)--(-.65,-.25);
			\draw[thick] (.65,1.25)--(1.25,1.25)--
				(2.5,1)--(2.5,0)--(1.25,-.25)--(.65,-.25);
			\draw[thick, fill=black] (-1.25,1.25) circle (0.1);
			\draw[thick, fill=white] (-2.5,1) circle (0.1);
			\draw[thick, fill=black] (-2.5,0) circle (0.1);
			\draw[thick, fill=white] (-1.25,-.25) circle (0.1);
			\draw[thick, fill=white] (1.25,1.25) circle (0.1);
			\draw[thick, fill=black] (2.5,1) circle (0.1);
			\draw[thick, fill=white] (2.5,0) circle (0.1);
			\draw[thick, fill=black] (1.25,-.25) circle (0.1);
			\node at (0,-.25) {$\ldots$};
			\node at (0,1.25) {$\ldots$};
		\end{scope}
	}
	
	\begin{scope}[shift={(-6,0.5)}]
		\draw[thick] (0:.75)--(-120:.75)--(120:.75)--cycle;
		
		\draw[thick, fill=white] (0:.75) circle (0.1);
		\draw[thick, fill=white] (-120:.75) circle (0.1);
		\draw[thick, fill=black] (120:.75) circle (0.1);
		
		\node at (120:1) [left] {$\bar u_1$};
		\node at (0:.75) [above right] {$u_2$};
		\node at (-120:1) [left] {$u_3$};
	\end{scope}
	
	\begin{scope}[shift={(6,0.5)}]
		\draw[thick] (180:.75)--(-60:.75)--(60:.75)--cycle;
	
		\draw[thick, fill=black] (180:.75) circle (0.1);
		\draw[thick, fill=white] (-60:.75) circle (0.1);
		\draw[thick, fill=white] (60:.75) circle (0.1);
		
		\node at (60:1) [right] {$u_1$};
		\node at (180:.75) [above left] {$\bar u_2$};
		\node at (-60:1) [right] {$u_3$};
	\end{scope}
	
	\node at (-2.5,-2.05) [below] {$\bar u_3$};	
	\node at (-2.5,-1) [below left] {$u_3$};
	\node at (-2.5,-.05) [below] {$\bar u_2$};	
	\node at (-2.5,1) [left] {$u_2$};
	\node at (-2.5,2) [above left] {$\bar u_1$};	
	\node at (-2.5,3) [above] {$u_1$};
	
	\node at (2.5,-2) [right] {$u_3$};	
	\node at (2.5,-1) [above] {$\bar u_3$};
	\node at (2.5,-.05) [above right] {$u_2$};	
	\node at (2.5,1) [right] {$\bar u_2$};
	\node at (2.5,2) [above right] {$u_1$};	
	\node at (2.5,3) [above] {$\bar u_1$};
	
	\node at (0,2.5) {\sffamily $\mathsf{u_1}$ variable gadget};
	\node at (0,0.5) {\sffamily $\mathsf{u_2}$ variable gadget};
	\node at (0,-1.5) {\sffamily $\mathsf{u_3}$ variable gadget};
	
	\node[rotate=90] at (-8,.5) {\sffamily $\mathsf{C_1}$ clause gadget};
	\node[rotate=-90] at (8,.5) {\sffamily $\mathsf{C_2}$ clause gadget};

	\end{tikzpicture}
\caption{The gadget graph $G_\varphi$ for 
$\varphi=(\bar u_1\vee u_2\vee u_3)\wedge(u_1\vee \bar u_2\vee u_3)$.}
\label{fig:3-sat-graph}
\end{figure}

The next lemma is equivalent to the result shown by Wang and Kuo~\cite{WK88}. 
\begin{lemma}\cite{WK88} The {\sc planar 3-sat} instance $\varphi$ is satisfiable 
if and only if $G_\varphi$ has vertex cover set of size at most 
$2m+ \frac12\sum_{i=1}^n {r_i}$,
where $r_i$ is the number of vertices in the variable gadget's cycle for $u_i$. 
\end{lemma}
\begin{proof}
	Suppose $\varphi$ is satisfiable. Take any satisfying assignment, and let 
	$U$ be the set of literals which take {\sc true} values in the assignment,
	$i.e.$ the literal ``$u_i$'' if the variable $u_i$ was assigned {\sc true}, 
	and the literal ``$\bar u_i$'' if the variable $u_i$ was assigned {\sc false}. 
	Let $\bar{U}$ be the remaining literals.
	Now, every vertex in a variable gadget of $G_\varphi$ whose label is in 
	$\bar{U}$ will be selected to be in the vertex cover. 
	In total this amounts to $\frac12\sum_{i=1}^n {r_i}$ vertices, 
	and these {\em cover} every edge in the variable gadgets of $G_\varphi$.
	Next consider the clause gadgets of $G_\varphi$. 
	We must select two vertices of each clause gadget to cover the edges of the 
	$3$-cycle. This amounts to $2m$ vertices. 
	Since we have a satisfying assignment and we chose the nodes 
	corresponding to $\bar{U}$ in the variable gadgets,
	each clause gadget must have at least one incident edge covered by the
	variable gadgets' selected vertices. 
	Hence, selecting the other two vertices will cover all incident edges
	to the clause gadgets, and all edges of the gadget's cycle.
	Thus we have a vertex cover with $2m+ \frac12\sum_{i=1}^n {r_i}$ nodes. 
	For example, in Figure~\ref{fig:3-sat-graph}, if we set all variables
	to {\sc false}, one possible vertex cover is the set of vertices labelled
	with non-negated literals, $i.e.$ those coloured in white. (Note, clearly, the set of white vertices
	will not typically form a vertex cover in the gadget graph.)
	
	Conversely, suppose we have a vertex cover $\mathcal{C}$ containing at 
	most $2m+\sum_{i=1}^n \tfrac{r_i}2$ vertices. 
	Each variable gadget must contribute at least $\tfrac{r_i}2$ vertices, 
	otherwise we cannot cover every edge in its cycle.
	Each clause gadget must contribute at least two vertices, or one edge in 
	the $3$-cycle will be uncovered. 
	Hence, $\mathcal{C}$ contains exactly $2m+\sum_{i=1}^n \tfrac{r_i}2$ 
	vertices.
	The $\frac12 r_i$ vertices from the variable gadget for $u_i$ corresponds 
	either to the set of all vertices with negated labels or to the set with 
	non-negated labels, otherwise there is an uncovered edge in the cycle. 
	This induces a truth assignment; set $u_i$ to {\sc true} if all the 
	``$\bar u_i$''-labelled vertices are selected, 
	and {\sc false} if the ``$u_i$''-labelled vertices are selected.
	Furthermore this is a satisfying assignment. 
	To see this note that as $\mathcal{C}$ covers all edges, the unselected 
	vertex in each clause is a literal which evaluates to {\sc true} by the
	selected assignment.
	\qed
\end{proof}

Hence, to solve for the satisfiability of $\varphi$, it suffices to test 
whether $G_\varphi$ admits a vertex cover with at most 
$2m+\tfrac12\sum_{i=1}^n r_i$ vertices. 
It remains to show the second of the three parts of the reduction. That is,
we need to show that this {\sc vertex cover} problem in the undirected 
gadget graph can be solved by finding a minimum directed feedback vertex set 
in an oriented-disc graph $D_\varphi$.
The basic idea is straightforward (albeit that the implementation is intricate). 
The oriented-disc graph $D_\varphi$ will contain a digon for each edge in some
$G_\varphi$. However, it will also contain a collection of additional arcs.
The key fact will be that these additional arcs form an acyclic subgraph 
of $D_\varphi$.
Thus every cycle in $D_\varphi$ must induce a digon. 
Consequently, a minimum directed feedback vertex set need only intersect 
each digon to ensure that every cycle is hit.
As argued previously, hitting the underlying graph formed by the 
digons of $D_\varphi$ 
corresponds to selecting a vertex cover in $G_\varphi$, as desired.
We now formalise this argument.

\begin{lemma}For every instance $\varphi$ of {\sc planar 3-sat}, 
there exists an oriented-disc graph $D_\varphi$ on which the 
{\sc directed feedback vertex set} problem is equivalent to the 
{\sc vertex cover} problem on $G_\varphi$.
\end{lemma}
\begin{proof}
We prove this by explicitly constructing the oriented-disc drawing. 
Recall the disc graph $D_\varphi$ should contain a digon for each edge in $G_\varphi$.
To do this, we begin with sufficiently a large planar drawing of $H_\varphi$, the planar 
bipartite network associated with $\varphi$.
At each clause vertex, we place an oriented-disc construction for the clause gadget.
This construction, along with its resulting 
graph, is shown in Figure \ref{fig:clause-gadget}. The figure
shows a clause gadget and a section of each of the neighbouring three
variable gadgets to which it is attached. Observe from the figure that, as claimed, 
the set of arcs created in $D_\varphi$ which 
are not in a digon, form an acyclic subgraph of $D_\varphi$.

\begin{figure}[p]\centering
\begin{tikzpicture}[scale=0.2]
	\clip (-29,-31) rectangle (45,21);
	\foreach \r in {0,1,2} {
		\begin{scope}[rotate=120*\r] 
			\coordinate (A) at (0,5);
			\coordinate (B) at (8.5,0);
			\coordinate (C) at (11,2.25);
			\coordinate (D) at (14,-2);
			\coordinate (E) at (17,2);
			\coordinate (F) at (20,-2);
			\coordinate (G) at (10,5);
			\coordinate (H) at (8.75,5.75);
			\coordinate (I) at (12,7);
			\coordinate (J) at (13,14);
			\coordinate (K) at (16,11);
			\coordinate (L) at (19,15);
			\coordinate (M) at (22,11);
						
			\oriendisc{(A)}{(5.5,-9.5)}
			\oriendisc{(B)}{(4,14)}
			\oriendisc{(C)}{(11,-5.75)}
			\oriendisc{(D)}{(14,5.5)}
			\oriendisc{(E)}{(17,-6)}
			\oriendisc{(F)}{(20,6)}
			
			\oriendisc{(G)}{(1,1)}
			\oriendisc{(H)}{(13.75,5.75)}
			\oriendisc{(I)}{(2,17)}
			\oriendisc{(J)}{(13,6)}
			\oriendisc{(K)}{(16,19)}
			\oriendisc{(L)}{(19,7)}
			\oriendisc{(M)}{(22,19)}
		\end{scope}
	}
	
	\begin{scope}[scale = 0.6,rotate = -75, shift={(38,45)}]
			\foreach \r in {0,1,2} {
			\begin{scope}[rotate=120*\r]
				\fill [color=black] (0,2) circle (0.7);
				\fill [color=black] (0,5) circle (0.7);
				
				\fill [color=black] (3.5,6.5) circle (0.7);
				\fill [color=black] (6,10) circle (0.7);
				\fill [color=black] (7,13) circle (0.7);
				\fill [color=black] (7,16) circle (0.7);
								
				\fill [color=black] (-3.5,6.5) circle (0.7);
				\fill [color=black] (-6,10) circle (0.7);
				\fill [color=black] (-7,13) circle (0.7);
				\fill [color=black] (-7,16) circle (0.7);
				\fill [color=black] (-7,19) circle (0.7);
				\fill [color=black] (-7,22) circle (0.7);
				
				\draw[thick] (0,2)--(-30:2);
				\draw[thick] (0,2)--(0,5);
				
				\draw[thick] (0,5)--(3.5,6.5);
				\draw[thick] (0,5)--(-3.5,6.5);

				\draw[thick] (3.5,6.5)--(6,10);
				\draw[thick] (-3.5,6.5)--(-6,10);

				\draw[thick] (6,10)--(7,13);
				\draw[thick] (-6,10)--(-7,13);

				\draw[thick] (7,13)--(7,16);
				\draw[thick] (-7,13)--(-7,16);
				
				\draw[thick] (-7,16)--(-7,19);
				\draw[thick] (-7,19)--(-7,22);
				
				\draw[-stealth, dotted, thick] (0,5) to[bend right=40] (-5,10);
				\draw[-stealth, dotted, thick] (0,5) to[bend right=50] (-6,13);
			\end{scope}
			}
		\end{scope}
\end{tikzpicture}	
\caption{Oriented-disc construction of the clause gadget, and its resulting graph.}
\label{fig:clause-gadget}

\vspace*{2em}

\subfloat[][Straight line]{\centering
	\begin{tikzpicture}
	\clip (-.625,-.875) rectangle (5.125,.875);
	\begin{scope}[scale = 0.2]
		\foreach \s in {0,1,...,6} {
			\oriendisc{(3.5*\s,1)}{(3.5*\s,-4)}
			\oriendisc{(3.5*\s+1.75,-1)}{(3.5*\s+1.75,4)}
		}
	\end{scope}
	\end{tikzpicture}
} 	\qquad \subfloat[][Curve]{\centering
	\begin{tikzpicture}
	\clip (-2.5,1.125) rectangle (2.5,3.1);
	\begin{scope}[scale = 0.2, rotate=45]
		\foreach \r in {0,1,...,6} {
			\oriendisc{(15*\r:10)}{(15*\r:15)}
		}
		\foreach \r in {0,1,...,5} {
			\oriendisc{(15*\r+7.5:12)}{(15*\r+7.5:8)}
		}
	\end{scope}
	\end{tikzpicture}
}
	
	\caption{Paths of bidirected edges as oriented-disc drawings.}
	\label{fig:path-graph}
	
\vspace*{2em}
	
	\includegraphics[height=5cm]{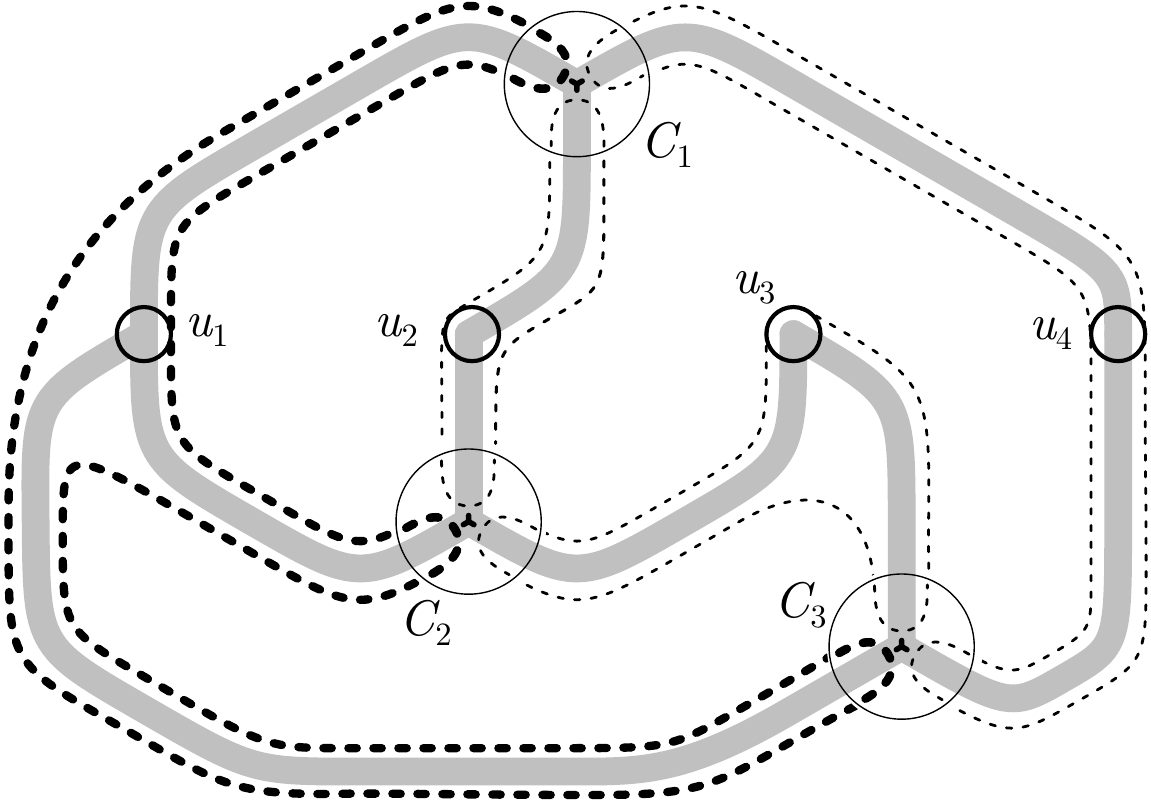}
	\caption{$G_\varphi$ as an oriented-disc graph.}
	\label{fig:3-sat-orien-graph}
\end{figure}

It remains to construct the large cycles for the variable gadgets, and connect
them to the clause gadgets. 
However, parts of these cycles are already included in the clause gadgets.
Thus, it suffices to join these cycle segments together via paths
of digons.
This can be done via the oriented disc constructions shown in 
Figure~\ref{fig:path-graph}.
To draw the cycle for some variable, say $u_i$, we note that $u_i$'s
vertex in the planar network $H_\varphi$ shares and edge with every clause
gadget which connects to $u_i$'s gadget.
Hence, as illustrated in Figure~\ref{fig:3-sat-orien-graph}, we may
follow along the edges of $H_\varphi$ to construct the cycle.
For example, in the figure, the variable cycle for $u_1$ (highlighted)
follows the topology of the edges incident to $u_1$'s vertex, and joins 
the clause gadgets (circled) to one another.

Observe that constructions in Figure~\ref{fig:path-graph} produce paths of digons 
in $D_\varphi$, where every arc produced is contained in a digon. 
It follows that the only arcs in $D_\varphi$ that are not in digons are in the
neighbourhoods of the clause gadgets and, as we have seen, these are acyclic.
But then, to hit all the cycles in $D_\varphi$, it suffices to hit all the digons, 
which, in turn, corresponds to a vertex cover in $G_\varphi$, completing the proof.
	\qed
\end{proof}
This completes all the steps in the reduction and we obtain:
	
	\begin{theorem} The {\sc consumer rationality} problem is NP-complete for a market with at 
	least 3 commodities. \qed
	\end{theorem}

 \end{document}